\documentclass[conference]{IEEEtran}
\usepackage[T1]{fontenc}
\usepackage{amsmath,amssymb,amsthm}
\usepackage{booktabs}
\usepackage{cite}
\usepackage{microtype}
\usepackage{balance}
\usepackage[hidelinks]{hyperref}
\newcommand{\Ing}{\operatorname{Ing}}
\newcommand{\cD}{\mathcal D}
\newcommand{\cR}{\mathcal R}
\newtheorem{theorem}{Theorem}
\newtheorem{proposition}{Proposition}
\hypersetup{pdftitle={Improved Bounds on the Ingleton Score of Four Random Variables},pdfauthor={Jiang Chang}}
\title{Improved Bounds on the Ingleton Score\\ of Four Random Variables}
\author{\IEEEauthorblockN{Jiang Chang}
\IEEEauthorblockA{College of Computer Science and Software Engineering\\Hohai University, Nanjing, China\\\texttt{science614@163.com}}}
\begin{document}
\maketitle

\begin{abstract}
For four finite discrete random variables, we prove that the negative of the Ingleton expression, divided by the joint entropy, is at most $19/124\approx0.15323$, improving the previously reported upper bound $3/19$ by $11/2356$. The proof uses a copy program with four copy steps and six auxiliary variables, certified by exact rational linear-programming certificates. We also show that two published information inequalities, due to Dougherty--Freiling--Zeger and Csirmaz, imply the bound $5/32$ by an explicit nonnegative combination with Shannon inequalities, and that $5/32$ is the exact optimum both of the union of the three published inequality catalogues considered here and of the full copy-lemma extensions of all $213$ distinct copy programs occurring in them. The new program is not among these programs, and the bound $19/124$ remains optimal when all of them and all catalogue inequalities are added. Ancillary files provide the parsed catalogues, exact primal and dual certificates, exact extension certificates, and verifiers that use only the Python standard library.
\end{abstract}


\section{Introduction}

Information inequalities describe the entropy region and constrain information-theoretic optimization problems, including network coding~\cite{ChanYeung2002,DFZ2007}. Ingleton's inequality holds for linear ranks~\cite{Ingleton1971}, whereas entropy vectors can violate it; the extent of this violation is one quantitative measure of the difference between these settings~\cite{MatusCsirmaz2016}.

For finite discrete random variables $A,B,C,D$, write
\begin{equation}
\begin{aligned}
\Ing={}&I(A;B\mid C)+I(A;B\mid D)\\
&+I(C;D)-I(A;B),
\end{aligned}
\label{eq:ingleton}
\end{equation}
and define the largest normalized violation by
\begin{equation}
M=\sup_{H(ABCD)>0}\frac{-\Ing}{H(ABCD)},
\label{eq:score}
\end{equation}
where the supremum is over all such joint distributions, with arbitrary finite alphabets. With this sign convention, a larger value represents a stronger violation. Dougherty, Freiling, and Zeger~\cite{DFZ2011} call $\Ing/H(ABCD)$ the Ingleton score, so that the infimum of their score is $-M$; Boston and Nan~\cite{BostonNan2020} use the opposite sign, under which $M$ is the supremum of the score.

Constructions of Mat\'{u}\v{s} and Csirmaz~\cite{MatusCsirmaz2016} and Boston and Nan~\cite{BostonNan2020} give normalized violations approximately $0.09243$ and $0.0925000777$, respectively. On the inequality side, Dougherty, Freiling, and Zeger~\cite{DFZ2011} established $M\le 3/19$. G\"urp\i nar and Romashchenko~\cite{GurpinarRomashchenko2019} described this as the best known bound in 2019, and Boston and Nan~\cite{BostonNan2020} still reported this upper bound in 2020.

We obtain
\begin{equation}
M\le\frac{19}{124}=0.153225\ldots,\qquad
\frac3{19}-\frac{19}{124}=\frac{11}{2356}.
\label{eq:result}
\end{equation}
The bound is proved in Section~\ref{sec:new} from a single copy program with four copy steps and six auxiliary variables. On the way, we observe that two inequalities already present in the published catalogues imply
\begin{equation}
M\le\frac5{32}=0.15625
\label{eq:result532}
\end{equation}
by a short identity (Section~\ref{sec:532}), and that $5/32$ is exactly what the three published catalogues together, and the full copy extensions of all copy programs occurring in them, can imply about this objective (Sections~\ref{sec:catalogue} and~\ref{sec:full}). To the best of our knowledge, no upper bound below $3/19$ has been reported before. The bound $5/32$ follows from inequalities published in~\cite{DFZ2011,Csirmaz2013}, whose contribution we fully acknowledge.

\section{Preliminaries}

Let $N=\{A,B,C,D\}$ and $h_S=H(S)$ for $S\subseteq N$, with $h_\varnothing=0$. An entropy vector consists of the fifteen coordinates $h_S$, $S\ne\varnothing$. The logarithm base is immaterial for the normalized score. The entropy-coordinate form of~\eqref{eq:ingleton} is
\begin{equation}
\begin{aligned}
\Ing={}&h_{AB}+h_{AC}+h_{AD}+h_{BC}+h_{BD}\\
&-h_A-h_B-h_{CD}-h_{ABC}-h_{ABD}.
\end{aligned}
\label{eq:entropy-form}
\end{equation}
The elemental Shannon inequalities are
\begin{equation}
\begin{aligned}
H(X\mid N\setminus\{X\})&\ge0 &&(X\in N),\\
I(X;Y\mid Z)&\ge0 &&(Z\subseteq N\setminus\{X,Y\}),
\end{aligned}
\label{eq:shannon}
\end{equation}
where $X\ne Y$ and each unordered pair is used once. There are $4+6\cdot4=28$ such inequalities. These inequalities are insufficient to characterize the entropy region~\cite{ZhangYeung1998}.

All inequalities considered below are homogeneous. Their linear programming outer bounds can therefore be optimized under $h_N=1$. The symmetry group of the objective that we use is
\begin{equation}
G=\{\mathrm{id},(AB),(CD),(AB)(CD)\}.
\label{eq:group}
\end{equation}
For a linear information expression $F$, $gF$ denotes the expression obtained by applying $g$ to its variable labels. Both $\Ing$ and $h_N$ are invariant under $G$.

\section{Two Published Information Inequalities}

The two inputs are
\begin{equation}
\begin{aligned}
\cD={}&2\Ing+2I(A;B\mid C)+2I(A;C\mid B)\\
&+I(B;C\mid A)+3I(C;D\mid A)\ge0
\end{aligned}
\label{eq:D}
\end{equation}
and
\begin{equation}
\begin{aligned}
\cR={}&2\Ing+I(A;B\mid C)+3I(A;C\mid B)\\
&+I(B;C\mid A)+3I(C;D\mid A)\ge0.
\end{aligned}
\label{eq:R}
\end{equation}
Inequality~\eqref{eq:D} is equation (50) of~\cite{DFZ2011}, also listed as entry 4 in its Section VIII; it is also entry 4 in Appendix B of~\cite{Csirmaz2013}. Inequality~\eqref{eq:R} is entry 3 in that appendix, and entry 1 in the appendix to the journal version~\cite{Csirmaz2016}.

The coefficient conventions in these sources differ. In the nine-term convention of~\cite{Csirmaz2013}, beginning with $\Ing$ and followed by its eight listed conditional mutual informations, the respective vectors are
\begin{equation}
\begin{aligned}
\cD &: (2,2,0,2,1,0,0,3,0),\\
\cR &: (2,1,0,3,1,0,0,3,0).
\end{aligned}
\label{eq:coefficients}
\end{equation}
In the convention of~\cite{DFZ2011}, where a multiple of $I(A;B)$ is on the left, the vector for~\eqref{eq:D} is $(2,4,2,1,2,0,0,2,3)$. Moving that term to the right and substituting~\eqref{eq:ingleton} gives exactly~\eqref{eq:D}. Thus~\eqref{eq:D} and~\eqref{eq:R} use one common entropy-coordinate convention.

For completeness, the copy strings attached to these inequalities are
\begin{equation}
\begin{aligned}
\cD &: \texttt{rs=bd:ac; t=c:adr},\\
\cR &: \texttt{rs=bd:ac; tu=cd:abs; v=b:acrst}.
\end{aligned}
\label{eq:copy-strings}
\end{equation}
Lowercase letters name the corresponding original or auxiliary variables. A step $Y=X:Z$ creates a new tuple $Y$ with the same joint law with $Z$ as the source tuple $X$, and, given $Z$, conditionally independent of all previously existing variables. Successive steps retain all earlier variables. A tuple copy preserves the joint source marginal, including every subtuple marginal.

We verified~\eqref{eq:D} and~\eqref{eq:R} directly from these standard copy constructions. The resulting systems have, respectively, 127 and 511 nonempty entropy coordinates, 679 and 4,617 elemental Shannon rows, and 22 and 71 copy equalities. Their ancillary certificates express $\cD$ and $\cR$ as nonnegative combinations of Shannon rows plus signed combinations of copy equalities. All coefficients are integers for $\cD$ and have denominators dividing 6 for $\cR$. These homogeneous identities establish the two inequalities for every original joint distribution.

\section{The Bound 5/32}
\label{sec:532}

\begin{theorem}\label{thm:main}
For every four finite discrete random variables,
\begin{equation}
64\Ing+10H(ABCD)\ge0.
\label{eq:theorem}
\end{equation}
Consequently, the supremum in~\eqref{eq:score} satisfies $M\le5/32$.
\end{theorem}

\begin{proof}
The following identity holds for every vector of the fifteen entropy coordinates:
\begin{equation}
\begin{aligned}
64\Ing+10h_N
={}&\sum_{g\in G}(2g\cD+g\cR)\\
&+10\sum_{X\in N}H(X\mid N\setminus\{X\})\\
&+30I(A;B\mid CD)+10I(C;D)\\
&+2I(C;D\mid A)+2I(C;D\mid B).
\end{aligned}
\label{eq:p532}
\end{equation}
To verify it, expand the conditional mutual informations into entropy coordinates and substitute~\eqref{eq:D} and~\eqref{eq:R}. The coefficient of $h_N$ is 10; the remaining coefficients are 64 times those in~\eqref{eq:entropy-form}. The ancillary verifier checks each of the fifteen coefficients exactly.

Every $g\cD$ and $g\cR$ is nonnegative by the corresponding published inequality with its variables renamed. The conditional entropies in the second line of~\eqref{eq:p532} and all mutual informations in its last two lines are nonnegative Shannon terms. Every scalar coefficient on the right is positive. This proves~\eqref{eq:theorem}. If $h_N>0$, division by $64h_N$ gives $-\Ing/h_N\le10/64=5/32$.

All four renamed expressions are evaluated on the same arbitrary joint distribution. No invariance of that distribution under $G$ is required. If $h_N=0$, all four variables are constant and~\eqref{eq:theorem} holds with equality.
\end{proof}

\section{Exact Optima of the Finite Catalogue Bounds}
\label{sec:catalogue}

For a finite collection $\mathcal C$ of published inequalities, let $P(\mathcal C)$ be the polyhedron defined by the 28 Shannon rows, all 24 variable permutations of every member of $\mathcal C$, and $h_N=1$. Write
\begin{equation}
U(\mathcal C)=\max_{h\in P(\mathcal C)}-\Ing(h).
\label{eq:catalogue-lp}
\end{equation}
Since every normalized entropy vector satisfies these constraints, $M\le U(\mathcal C)$.

We use the 214 inequalities tabulated in~\cite{DFZ2011}, the 435 entries of Appendix B in the 2013 arXiv version~\cite{Csirmaz2013}, and the 264 entries of Appendix B in the journal version~\cite{Csirmaz2016}. The last list contains newly listed inequalities and is a different catalogue, rather than a cumulative replacement for the 2013 list. Table~\ref{tab:optima} gives the exact optima. The numbers of rows retain repeated permutations and include the 28 Shannon rows.

\begin{table}[ht]
\caption{Exact finite-catalogue upper bounds}
\label{tab:optima}
\centering
\begin{tabular}{@{}lrr@{}}
\toprule
Catalogue & Rows & $U(\mathcal C)$\\
\midrule
DFZ, 214 entries & 5,164 & $3/19$\\
Csirmaz 2016, 264 entries & 6,364 & $59/376$\\
Csirmaz 2013, 435 entries & 10,468 & $5/32$\\
DFZ + Csirmaz 2016 & 11,500 & $5/32$\\
\bottomrule
\end{tabular}
\end{table}

\begin{proposition}
The optimum is $5/32$ for the Csirmaz 2013 catalogue, for the union of DFZ with Csirmaz 2016, and for the union of all three catalogues. It is also $5/32$ for the smaller system consisting of the 28 Shannon rows and the four $G$-images of each of $\cD$ and $\cR$.
\end{proposition}

\begin{proof}
Consider the rational vector $h^*$ defined by
\begin{equation}
\begin{array}{c|c}
S & h^*_S\\\hline
A,B & 9/16\\
C,D & 1/2\\
AB & 27/32\\
AC,AD,BC,BD & 25/32\\
CD,ABC,ABD,ACD,BCD,ABCD & 1
\end{array}
\label{eq:primal}
\end{equation}
It has $h_N^*=1$ and $\Ing(h^*)=-5/32$. Direct rational evaluation shows that it satisfies every Shannon row and every permutation of every inequality in all three catalogues. This is a finite check of 21,940 rows when the three catalogues are combined; the ancillary files contain the complete coefficient data and verifier. Thus each of these polyhedra has objective value at least $5/32$.

For the 2013 catalogue, identity~\eqref{eq:p532} supplies the matching upper bound because both $\cD$ and $\cR$ occur in that list. For DFZ combined with the 2016 list, $\cD$ occurs in DFZ and $\cR$ occurs in the 2016 list. The same identity applies to their union and to the stated smaller system. Hence all these optima are exactly $5/32$.
\end{proof}

For the other two rows of Table~\ref{tab:optima}, the ancillary files give exact rational primal and dual solutions. In the convention $q_i(h)\ge0$ for constraint rows, each dual certificate has the form
\begin{equation}
\Ing=\nu h_N+\sum_i\lambda_i q_i,
\qquad \lambda_i\ge0,
\label{eq:catalogue-dual}
\end{equation}
with $\nu=-3/19$ or $-59/376$, respectively. A feasible primal vector attains $\Ing=\nu$ in each case. Exact feasibility and coordinate equality therefore establish both bounds and their optimality. The same certificate format is supplied for all five catalogue combinations, including the union of all three.

\setlength{\splitmaxdepth}{0pt}
\balance
\section{Full Copy Extensions of the Catalogued Programs}
\label{sec:full}

A copy step creates a new tuple $F$ from a source tuple $X$ over a conditioning tuple $Z$, where $X$ and $Z$ consist of existing variables and are disjoint; parenthesized groups in a copy string denote a single new variable copying a tuple~\cite{DFZ2011,Csirmaz2013}. Writing $O$ for the set of variables existing before the step and indexing the components of $X$ and $F$ by a common set $T$, the step contributes the equalities
\begin{align}
H(F_S W)&=H(X_S W), \label{eq:C1}\\
I(F;O\setminus Z\mid Z)&=0, \label{eq:C2}\\
H(F_S X_{S'} W)&=H(X_S F_{S'} W), \label{eq:C3}
\end{align}
for all $S,S'\subseteq T$, $S\neq\varnothing$ in~\eqref{eq:C1}, and all $W\subseteq Z$. Equalities~\eqref{eq:C1} and~\eqref{eq:C2} hold for the standard construction $p(o,f)=p(o)\,p_{X\mid Z}(f\mid z)$. Under this construction, $X$ and $F$ are conditionally independent and identically distributed given $Z$, so the pair $(X,F)$ is exchangeable given $Z$; this gives~\eqref{eq:C3}.

For a copy program, its \emph{block} consists of the entropy coordinates of all nonempty subsets of the original and new variables, the elemental Shannon inequalities on these variables, and the equalities~\eqref{eq:C1}--\eqref{eq:C3} of every step. Every entropy vector of four random variables extends to a vector satisfying the block, namely the entropy vector of the constructed joint distribution. Each of the four $G$-images of a program receives its own block; different blocks share only the fifteen original coordinates.

After renaming new variables, the three catalogues contain $213$ distinct copy programs, with at most four steps and five new variables.

\begin{proposition}
Add to the polyhedron $P(\mathcal C)$ of the union of the three catalogues the blocks of all four $G$-images of all $213$ programs, that is, $852$ blocks. The optimum of $-\Ing$ over the resulting polyhedron is $5/32$.
\end{proposition}

\begin{proof}
Theorem~\ref{thm:main} gives the upper bound. The ancillary files contain, for every program and every $G$-image, an exact rational vector on the variables of its block whose fifteen original coordinates are those of~\eqref{eq:primal} and which satisfies every row of the block. Since the blocks share only the original coordinates, these vectors combine into one feasible point with objective value $5/32$.
\end{proof}

Consequently, any valid constraint that improves the bound below $5/32$ is violated by the vector~\eqref{eq:primal}. No published copy program, used as a full extension with the exchange equalities~\eqref{eq:C3} and all four $G$-images, provides such a constraint.

\section{A New Copy Program}
\label{sec:new}

Consider the copy program
\begin{equation}
\begin{aligned}
\mathcal P:\quad &\texttt{rs=bd:ac; tu=cd:abs;}\\
&\texttt{v=b:acrst; w=b:crstv}.
\end{aligned}
\label{eq:newprogram}
\end{equation}
Its first three steps form the copy program of $\cR$ in~\eqref{eq:copy-strings}; the last step copies $b$ over the five variables $c,r,s,t,v$. It has four steps and six new variables, so its block has ten variables and $1023$ nonempty entropy coordinates. It is not among the $213$ catalogued programs, which use at most five new variables. It was found by a systematic search over local modifications of the catalogued programs, in which one step is added, deleted or changed.

\begin{theorem}\label{thm:new}
For every four finite discrete random variables,
\begin{equation}
124\Ing+19H(ABCD)\ge0.
\end{equation}
Consequently, the supremum in~\eqref{eq:score} satisfies $M\le19/124$.
\end{theorem}

\begin{proof}
Let $\Pi$ be the polyhedron formed, as in Section~\ref{sec:full}, by the blocks of the four $G$-images of $\mathcal P$, which share the fifteen original coordinates; it has $4047$ coordinates. The entropy vector of $A,B,C,D$ extends to a point of $\Pi$, namely to the entropy vectors of the four joint distributions constructed by the copy steps. The ancillary file \texttt{new\_program\_19\_124.json} gives rational multipliers $\lambda_i\ge0$ on $1{,}862$ elemental Shannon rows $s_i$ and $\mu_j$ on $380$ copy equalities $e_j$ of these blocks such that
\begin{equation}
\Ing=-\frac{19}{124}\,h_N+\sum_i\lambda_i s_i+\sum_j\mu_j e_j
\label{eq:newdual}
\end{equation}
holds identically in all $4047$ coordinates. Evaluating~\eqref{eq:newdual} at an extension of the entropy vector gives $\Ing\ge-\frac{19}{124}h_N$.
\end{proof}

The same ancillary file contains a rational point of $\Pi$ with $h_N=1$ and $-\Ing=19/124$. Its original coordinates are
\begin{equation}
\begin{aligned}
&h'_A=h'_B=69/124,\qquad h'_C=h'_D=1/2,\\
&h'_{AB}=105/124,\\
&h'_{AC}=h'_{AD}=h'_{BC}=h'_{BD}=193/248,
\end{aligned}
\label{eq:newpoint}
\end{equation}
with all remaining original coordinates equal to $1$. Hence $19/124$ is the exact optimum of $\Pi$. The vector~\eqref{eq:newpoint} satisfies every permutation of every catalogue inequality, and it extends to each of the $852$ blocks of Section~\ref{sec:full}. Therefore $19/124$ is also the exact optimum of the relaxation that combines $\Pi$ with all catalogue inequalities and all $852$ blocks, and any valid constraint that improves the bound below $19/124$ is violated by~\eqref{eq:newpoint}.

\section{Discussion}
\label{sec:discussion}

Keeping a full copy extension can give a stronger outer bound than keeping only its displayed consequences. The copy program of $\cR$ in~\eqref{eq:copy-strings}, used as a full extension with its four $G$-images, gives $5/32$ by itself, whereas the four $G$-images of any single catalogue inequality attached to this program give at best $1/6$. For entry 219 of~\cite{Csirmaz2013}, the corresponding values are $11/70$ and $19/120$. Exact certificates for these comparisons are included in the ancillary files.

The entropy-region problem remains separated from the finite catalogue optima. The construction of~\cite{BostonNan2020} and Theorem~\ref{thm:new} place $M$ in the interval from approximately $0.0925000777$ to $19/124\approx0.153226$. The vectors~\eqref{eq:primal} and~\eqref{eq:newpoint} certify optimality of the stated polyhedral bounds; a corresponding lower bound on $M$ would require a distributional construction.

\section{Reproducibility}
\label{sec:repro}

The \texttt{anc/} directory contains \texttt{verify\_5\_32.py}, the three parsed catalogue JSON files, their exact primal and dual certificates, and the two copy certificates for $\cD$ and $\cR$. It also contains the list of the $213$ copy programs, the exact extension vectors of~\eqref{eq:primal} and~\eqref{eq:newpoint} for all $852$ blocks, the certificates for Theorem~\ref{thm:new} and for the comparisons in Section~\ref{sec:discussion}, and \texttt{verify\_full\_extensions.py}. From the source-package root, run
\begin{center}
\texttt{python3 anc/verify\_5\_32.py}\\
\texttt{python3 anc/verify\_full\_extensions.py}.
\end{center}
The verifier performs all arithmetic with the standard-library \texttt{Fraction} class. It expands~\eqref{eq:p532}, tests~\eqref{eq:primal} against every catalogue row, reconstructs the copy constraints from~\eqref{eq:copy-strings}, and verifies the exact primal and dual certificates. The second verifier regenerates the elemental Shannon inequalities and the equalities~\eqref{eq:C1}--\eqref{eq:C3} of every block from the copy strings and checks every extension vector, the certificates of Theorem~\ref{thm:new}, and the comparison certificates exactly. Both run without a numerical solver or network access. The README files specify coefficient bases, row order, source versions, file formats, and the commands for the comparisons.

\section*{Acknowledgment}

The AI systems ChatGPT (ChatGPT~6.0), Codex (\mbox{gpt-6-astra}) and Claude (Claude Opus~5.5) were used for literature search, for deriving and checking the combination of known inequalities, for the systematic search over copy programs, for writing the exact-arithmetic verification code and running the linear-programming computations, and for drafting the manuscript. All results reported here are confirmed by the exact checks described in Section~\ref{sec:repro}. The author takes full responsibility for the content.

\clearpage
\nobalance
\IEEEtriggeratref{6}
\bibliographystyle{IEEEtran}
\bibliography{refs}
\end{document}